\documentclass[twosided,11pt]{article}      
\usepackage{amsmath,amsfonts} 
\usepackage[pdftex]{graphicx}
\usepackage{amsthm}
\usepackage[left=1in,top=1in,bottom=1in,right=1in]{geometry}
\usepackage{natbib}
\usepackage{algorithm}
\usepackage{algorithmic}
\usepackage{enumerate}
\usepackage{setspace}
\usepackage{multirow}

\newcommand{\R}{\mathbb{R}}
\newcommand{\p}{\mathbb{P}}
\newcommand{\E}{\mathbb{E}}
\newcommand{\1}{\mathbf{1}}

\newcommand{\g}{\, | \,}

\newcommand{\x}{\mathbf{x}}
\newcommand{\X}{\mathbf{X}}
\newcommand{\y}{\mathbf{y}}

\newcommand{\A}{\mathbf{A}}

\newtheorem{thm}{Theorem}[section]

\newtheorem{lem}[thm]{Lemma}

\begin{document}

\title{Bayesian nonparametric multivariate convex regression}

\author{Lauren Hannah, David Dunson}

\maketitle

\begin{abstract}
In many applications, such as economics, operations research and reinforcement learning, one often needs to estimate a multivariate regression function $f$ subject to a convexity constraint.  For example, in sequential decision processes the value of a state under optimal subsequent decisions may be known to be convex or concave. We propose a new Bayesian nonparametric multivariate approach based on characterizing the unknown regression function as the max of a random collection of unknown hyperplanes.  This specification induces a prior with large support in a Kullback-Leibler sense on the space of convex functions, while also leading to strong posterior consistency. Although we assume that $f$ is defined over $\R^p$, we show that this model has a convergence rate of $\log(n)^{-1} n^{-1/(d+2)}$ under the empirical $L_2$ norm when $f$ actually maps a $d$ dimensional linear subspace to $\R$. We design an efficient reversible jump MCMC algorithm for posterior computation and demonstrate the methods through application to value function approximation.
\end{abstract}


\section{Introduction}
Consider the problem of estimating the function $f$ for the model
$$y = f(\x) + \epsilon,$$
where $\x \in \mathcal{X} \subset \R^p$, $y \in \R$, $f:\R^p \rightarrow \R$ is a mean regression function and $\epsilon\sim N(0,\sigma^2).$ Given the observations $(\x_1,y_1),\dots,(\x_n,y_n)$, we would like to estimate $f$ subject to the convexity constraint, 
\begin{equation}\label{eq:convexity}
f(\x_1) \geq f(\x_2) + \nabla f(\x_1)^T (\x_1-\x_2),
\end{equation} for every $\x_1,\x_2 \in \mathcal{X}$, where $\nabla f(\x)$ is the gradient of $f$ at $\x$. This is called the convex regression problem. Convex regression can easily be modified to allow concave regression by multiplying all of the values by negative one.

Convex regression problems are common in economics, operations research and reinforcement learning. In economics, production functions~\citep{Sk78} and consumer preferences~\citep{MePr68} are often convex, while in operations research and reinforcement learning, value functions for stochastic optimization problems can be convex~\citep{ShDeRu09}. If a problem is known to be convex, a convex regression estimate provides advantages over an unrestricted estimate. First, convexity is a powerful regularizer:  it places strong conditions on the derivatives---and hence smoothness---of a function. Convexity constraints can substantially reduce overfitting and lead to more accurate predictions. Second, maintaining convexity allows the use of convex optimization solvers when the regression estimate is used in an objective function of an optimization problem.

Multivariate convex regression has received relatively little attention in the literature. The oldest method is the least squares estimator (LSE)~\citep{Hi54,Dy83,BoVa04,SeSe11},
\begin{align}\label{eq:old}
\min_{\hat{y}_{1:n}, \mathbf{g}_{1:n}} & \sum_{i=1}^n \left(y_i - \hat{y}_i\right)^2 \\\notag
\mathrm{subject \ to \ } & \hat{y}_j \geq \hat{y}_i + \mathbf{g}_i^T(\x_j - \x_i), \ \ \ i,j = 1,\dots,n.
\end{align}The resulting function is piecewise linear, generated by taking the maximum over the supporting hyperplanes, $\mathbf{g}_{1:n}$. However, Equation (\ref{eq:old}) has $n^2$ constraints, making solution infeasible for more than a few thousand observations. Recently, there has been interest in multivariate convex regression beyond the LSE. \citet{HePa09} proposed a method that generates a regression estimator via a weighted kernel estimate subject to conditions on the Hessian of the estimator; solutions are found using sequential quadratic programming. Convexity is guaranteed only at points where the Hessian condition is enforced and the method does not scale well to high dimensions or large datasets. \citet{HaDu11b} proposed a method, Convex Adaptive Partitioning (CAP), that adaptively splits the dataset and fits linear estimates within each of the subsets. Like the least squares estimator, the CAP estimator is formed by taking the maximum over hyperplanes; unlike previous methods, it produces a sparse estimator that scales well to large datasets and large numbers of covariates. However, it has theoretical guarantees only in the univariate case.

Piecewise planar models, like the LSE and CAP, are poor when used in the objective function of an optimization problem. The minima of piecewise planar functions occur at a vertex where $p+1$ hyperplanes intersect. The location of vertices is sensitive to the number of hyperplanes and the hyperplane parameters. The parameters are in turn sensitive to noise and observation design. Bayesian models could reduce these problems: prior distributions on parameters reduce design sensitivity and model averaging produces a smoother estimate.

Bayesian models have been used for convex regression, but only in the univariate case. In this setting, methods rely on the ordering implicit to the real line:  a positive semi-definite Hessian translates into an increasing derivative function in one dimension. \citet{RaLaSm93} discretized the covariate space and placed a Dirichlet prior over the normalized integral of the slope parameters between those points. \citet{ChChLi07} used Bernstein polynomials as a basis by placing a prior on the number of polynomials and then sampling from a restricted set of coefficients. \citet{ShWaDa11} used fixed knot and free-knot splines with a prior that placed an order restriction on the coefficients for each basis function. In a single dimension, Bayesian convex regression is closely related to Bayesian isotonic regression~\citep{LaMo95,NeDu04,ShSaWa09}. In multiple dimensions, however, convexity constraints become combinatorially difficult to enforce through projections.




We take an entirely different approach to modeling convex functions. Instead of creating an estimator based on a set of restricted parameters or projecting an unconstrained estimate back into the space of convex functions, we place a prior over a smaller set of functions that are guaranteed to be convex:  piecewise planar functions. The number of hyperplanes and their parameters are random; we define the function to be the maximum over the set of hyperplanes. We efficiently sample from the posterior distribution with reversible jump Markov chain Monte Carlo (RJMCMC). We call this approach Multivariate Bayesian Convex Regression (MBCR). Although the set of piecewise planar functions does not include all convex functions, it is dense over that space and we show strong ($L_1$) consistency for MBCR. If $f(\x) = g(\A\x)$ for some $d \times p$ matrix $\A $ and function $g$, we show convergence rates for MBCR with respect to the $L_2$ norm to be $\log(n)^{-1} n^{-1/(d+2)}$. The dimension of the linear subspace, $d$, determines the convergence rate, not the dimension of the full space, $p$.


In numerical experiments, we show that MBCR produces estimates that are competitive with LSE and CAP in terms of traditional metrics, like mean squared error, and can outperform them in objective function approximation. Through examples on toy problems, we show that MBCR has the potential to produce regression estimates that are much better suited to objective function approximation than piecewise planar methods. 






\section{Multivariate Bayesian Convex Regression}\label{sec:model}

Convexity is defined by the set of supporting hyperplane constraints in Equation (\ref{eq:convexity}):  any supporting hyperplane of the function $f$ at $\x_1$ is less than or equal to $f(\x_2)$ at any other point $\x_2$. This is equivalent to $f$ having a positive semi-definite Hessian. In multiple dimensions, it is difficult to project onto the set of functions that satisfy these constraints. Instead of placing a prior over an unconstrained set of functions and then restricting the parameters to meet convexity conditions, we place a prior over a smaller set of functions that automatically meet the conditions. Specifically, for all $\x$ in a compact set $ \mathcal{X}$ we place a prior over all functions that are the maximum over a set of $K$ hyperplanes, $(\alpha_1,\beta_1),\dots,(\alpha_K,\beta_K)\in \R^{p+1}$,
\begin{equation}\label{eq:max}
f(\x) = \max_{k \in \{1,\dots,K\} } \alpha_k + \beta^T_k \x,
\end{equation}where $K$ is unknown. This set of functions can approximate any convex function $f$ arbitrarily well while maintaining straightforward inference.


Assuming $f(\x)$ follows Equation (\ref{eq:max}), we let
\begin{equation}\label{eq:fTheta}
Y_i = f(\x_i ; \theta) + \epsilon_i , \quad \epsilon_i  \sim N(0,\sigma^2),
\end{equation}where the unknown parameters are $$\theta =\{ K, \alpha = (\alpha_1,\dots,\alpha_K)^T, \beta = (\beta_1,\dots, \beta_K)^T, \sigma^2 \}.$$ The prior $\Pi$ over $\{K,\alpha,\beta, \sigma^2\}$ is factored as,
\begin{equation}\notag
\Pi(K,\alpha,\beta,\sigma^2) = \Pi_{\sigma}(\sigma^2) \Pi_K(K) \prod_{k=1}^K \Pi_{\theta}(\alpha_{k},\beta_k).
\end{equation}The prior for the variance parameter, $\sigma^2$, is defined as $\Pi_\sigma$, and the prior for the number of hyperplanes, $K$, is $\Pi_K.$ The hyperplane parameters, $\theta_k = (\alpha_k,\beta_k)^,$ are given the prior $\Pi_{\theta}.$ These yield the model,
\begin{align}\notag
K & \sim \Pi_K, & 
\sigma^2 & \sim \Pi_{\sigma}, & 
\theta_k \g K & \sim \Pi_{\theta}, & k = 1,\dots,K.
\end{align}

MBCR is similar to Bayesian adaptive regression spline (BARS) models~\citep{DeMaSm98,DiGeKa01,ShSaWa09,ShWaDa11} in that the method places a prior over a finite set of locally parametric models, with the prior accommodating uncertainty in the number of models, their locations and their parameters. Indeed, we use the same inference method:  reversible jump Markov chain Monte Carlo (RJMCMC). In both cases, RJMCMC works by adaptively adding and removing local models while updating the model-specific parameters. However, while BARS explicitly introduces random changepoints or knots within a region, in MBCR regions are implicitly defined as corresponding to locations across which a particular hyperplane dominates. Let $\{A_1,\dots,A_K\}$ be a partition of $\mathcal{X}$ where $$A_k = \{ \x \in \mathcal{X} \, : \, k = \arg \max_{j \in \{1,\dots, K\} } \alpha_j + \beta_j^T \x \}.$$As in the local knot search of \citet{DiGeKa01}, we use these regions to produce an efficient proposal distribution for the RJMCMC. We discuss implementation details for MBCR in Section \ref{sec:implementation}, but first we show consistency and rate of convergence for MBCR in Section \ref{sec:theory}.

\section{Theoretical Results}\label{sec:theory}
Posterior consistency occurs if the posterior assigns probability converging to one in arbitrarily small neighborhoods of the true function $f_0$ as the number of samples $n$ grows. The rate of convergence is the rate at which the neighborhood size can contract with respect to $n$ while still maintaining consistency. Despite the longstanding interest in shape-restricted estimators, relatively little work has explored their asymptotic properties---particularly in multivariate and Bayesian settings. In the frequentist framework, \citet{HaPl76} showed consistency of the univariate LSE for convex regression; \citet{GrJoWe01} showed it has a local convergence rate of $n^{-2/5}$. More recently, \citet{SeSe11} showed consistency for the multivariate LSE. 

There is also a recent literature on the related topic of multivariate convex-transformed density estimation. \citet{CuSaSt10} showed consistency for the MLE log-concave density estimator;  \citet{SeWe10} showed consistency for the MLE of convex-transformed density estimators and gave a lower minimax bound on the convergence rate of $n^{-2/(p+4)}$. Bayesian shape-restricted asymptotics have received even less attention. \citet{ShSaWa09} showed consistency for monotone regression estimation with free knot splines in the univariate case; this was extended to univariate convex regression estimation by \citet{ShWaDa11}.


 Let $\theta \in \Theta$ be the set of parameters to be estimated. Let $\Pi$ be the prior induced on $f$ by
 \begin{align}\label{eq:modelThm}
 K-1 & \sim Poisson(\lambda), & 
\sigma^2 & \sim \Pi_{\sigma}, & 
\theta_k \g K & \sim \Pi_{\theta}, & k = 1,\dots,K,
 \end{align}where $\Pi_{\sigma}$ is defined in Assumption {\bf B2} and $\Pi_{\theta}$ in Assumptions {\bf B3} and {\bf B4}.
 We consider strong, or $L_1$, consistency. That is, let $$L_{\epsilon} = \left\{ (f,\sigma) : \int_{\mathcal{X}} \left| f(\x) - f_0(\x) \right| dx < \epsilon, \ \left| \frac{\sigma}{\sigma_0} - 1\right| < \epsilon\right\},$$ where the data-generating model is 
 \begin{equation}\notag
 Y_i = f_0(\x_i) + \epsilon_i, \quad \epsilon_i  \sim N(0,\sigma_0)^2.
 \end{equation} We would like $\Pi(L_{\epsilon}^C | (X_i,Y_i)_{i=1}^n ) \rightarrow 0$ as $n \rightarrow \infty$, almost surely $\p_{f_0,\sigma_0}^{\infty}$, where $\p_{f_0,\sigma_0}^{\infty}$ is the product measure under the true distribution. Throughout the rest of this paper, we use lower case $\x_i$ and $y_i$ to denote known or observed quantities, while $\X_i$ and $Y_i$ denote random variables. We show that MBCR is strongly consistent under a general set of conditions.

Bayesian rates of convergence are slightly different from their frequentist counterparts. A series $(\epsilon_n)_{n=1}^{\infty}$ where $\epsilon_n \rightarrow 0 $ is a rate of convergence under a metric $d(\theta,\theta_0)$ if 
\begin{equation}\notag
\p_{f_0,\sigma_0}^{\infty} \, \Pi(\theta \in \Theta \, : \, d(\theta,\theta_0 ) \geq H_n \epsilon_n  | (X_i,Y_i)_{i=1}^n ) \rightarrow 0
\end{equation} for every $H_n \rightarrow \infty$. We examine convergence rates with respect to the empirical $L_2$ norm. Moreover, if $f_0$ actually maps a $d$-dimensional linear subspace of $\R^p$ to $\R$, then the convergence rate is determined by the dimensionality of the subspace, $d$, rather than the full dimensionality, $p$.


\subsection{Consistency}\label{sec:consistency}

We consider two design cases for consistency: fixed design and random design. We place a series of assumptions on the true function, the prior and the design. Some of the assumptions on the prior are specific to the design type. In both cases, we assume that $f_0$ is uniformly bounded:
\begin{enumerate}
	\item[{\bf B1.}] The function $f_0$ is uniformly bounded on the compact set $\mathcal{X}$.
\end{enumerate}Without loss of generality, we assume that $\mathcal{X} = [0,1]^p$.

For both design types, we need define the prior $\Pi_{\sigma}$ and $\Pi_{\theta}$ in Equation (\ref{eq:modelThm}). First, we assume that the prior on $\sigma^2$ has compact support bounded away from zero. This is not a restrictive assumption in practice since zero measurement error is unlikely to occur and an upper bound can be easily chosen to cover a wide range of plausible values. Second, in the case of fixed design, we assume compact support of the prior for the hyperplane parameters; again, a wide range of plausible values can be chosen. Truncated normal and inverse-gamma distributions provide a convenient choice.
\begin{enumerate}
	\item[{\bf B2.}] Let $\Pi_{\sigma}$ be the prior on $\sigma$; $\Pi_{\sigma}$ is non-atomic and only has support over $[\underline{\sigma}, \bar{\sigma}]$ with $0 < \underline{\sigma} < \sigma_0 < \bar{\sigma}< \infty.$
	\item[{\bf B3.}] Let $\Pi_{\theta} = N_{p+1}(\mu_{\alpha,\beta},V_{\alpha, \beta})$ be the prior on $\theta_k$, where $N_{p+1}$ is the $p+1$ dimensional Gaussian distribution.
	 \item[{\bf B4.}] Let $\Pi^*_{\theta} = N_{p+1}(\mu_{\alpha,\beta},V_{\alpha, \beta})$. Let $L$ be a constant such that $L > || \frac{\partial}{\partial x_j} f_0(\x) ||_{\infty}$ and for some $V > \frac{1}{\sqrt{p}}L$, let $$ \Omega = \left\{ (\alpha, \beta) \, : \, \max \{\alpha,\beta_1,\dots,\beta_p\} \leq V\right\}.$$ Set $\Pi_{\theta} = \Pi^*_{\theta}(\cdot \cap \Omega) / \Pi^*_{\theta}(\Omega)$ and let $\theta_k \sim \Pi_{\theta}$.
\end{enumerate}

For both design cases, we need to ensure that the covariate space is sufficiently well-sampled.
\begin{enumerate}
	\item[{\bf B5.}] For each hypercube $H$ in $\mathcal{X}$, let $\lambda(H)$ be the Lebesgue measure. Suppose that there exists a constant $K_p$ with $0 < K_p \leq 1$ such that whenever $\lambda(H) \geq \frac{\lambda(\mathcal{X})}{K_p n}$, $H$ contains at least one design point for sufficiently large $n$.
	\item[{\bf B6.}] Let $Q$ be the density of the random design points; $Q$ is non-atomic and $Q(x) > 0$ for every $x \in \mathcal{X}$.
\end{enumerate}

With these assumptions, we now give consistency results.
\begin{thm}\label{thm:L1Fixed}
Assume that $\mathcal{X}$ is compact, the covariate design is fixed and that $f_0$ is convex with continuous first order partial derivatives. Suppose that conditions {\bf B1}, {\bf B2}, {\bf B4} and {\bf B5} hold. Then for every $\epsilon > 0$, $$ \p_{f_0,\sigma_0}^{\infty} \, \Pi \left( L_{\epsilon}^C \g Y_1,\dots, Y_n , \x_1, \dots, \x_n \right) \rightarrow 0.$$
\end{thm}

In the stochastic design case, assumptions {\bf B4} and {\bf B5} are replaced by {\bf B3} and {\bf B6}, respectively. We note that for random design, $L_1$ convergence follows directly from convergence in probability for a uniformly bounded function $f_0$.

\begin{thm}\label{thm:L1Random}
Assume that $\mathcal{X}$ is compact, the covariate design is random and that $f_0$ is convex with continuous first order partial derivatives. Suppose that conditions {\bf B1}-{\bf B3} and {\bf B6} hold. Then for every $\epsilon > 0$, $$\p_{f_0,\sigma_0}^{\infty}\, \Pi \left( L_{\epsilon}^C \g Y_1,\dots, Y_n , \X_1, \dots, \X_n \right) \rightarrow 0.$$
\end{thm}

To prove Theorems \ref{thm:L1Fixed} and \ref{thm:L1Random}, we use the consistency results for Bayesian nonparametric regression of \citet{ChSc07}. We show that the prior $\Pi$ satisfies the following assumptions,
\begin{enumerate}
	\item[{\bf A1.}] The prior $\Pi$ puts positive measure on the neighborhood $$B_{\delta} = \left\{ (f,\sigma) : || f - f_0 ||_{\infty} < \delta, \, \left| \frac{\sigma}{\sigma_0} - 1\right| < \delta \right\}$$ for every $\delta > 0$.
	\item[{\bf A2.}] Set $\Theta_n = \Theta_{1n} \times \R_+$, where $$\Theta_{1n} = \left\{f : ||f||_{\infty} < M_n, \, \left| \left| \frac{\partial}{\partial x_j} f \right| \right|_{\infty} < M_n, \, j = 1,\dots,p\right\},$$ and $M_n = \mathcal{O}(n^{\alpha})$ with $\frac{1}{2} < \alpha < 1$. Then there exists $C_1,c_1 > 0$ such that $\Pi(\Theta^C_n) \leq C_1 e^{-c_1n}$.
\end{enumerate}

\citet{ChSc07} modifies the consistency theorem of \citet{Sc65} for non-i.i.d. observations; the requirements are prior positivity on a set of variationally close Kullback-Leibler neighborhoods and the existence of exponentially consistent tests separating the desired posterior region from the rest. Assumption {\bf A1} satisfies the prior positivity while assumption {\bf A2} constructs a sieve that is used to create the exponentially consistent tests. Assumptions {\bf A1} and {\bf A2} generate pointwise convergence in the empirical and in-$Q$-probability metrics for the fixed and random design cases, respectively. Assumptions {\bf B1} to {\bf B6} are then used to extend consistency under these metrics to consistency under the $L_1$ metric. See Appendix A for details.	

\subsection{Rate of Convergence}\label{sec:rate}
We determine the rate of convergence of MBCR with respect to the empirical $L_2$ norm,
\begin{equation}\notag
||f||_n = \left(\frac{1}{n} \sum_{i=1}^n f(\x_i)^2 \right)^{1/2}.
\end{equation}For both the fixed design and random design cases, we make the following assumptions:
\begin{enumerate}
	\item[{\bf B7.}] The model variance $\sigma_0^2$ is known.
	\item[{\bf B8.}] There exists a convex function $g_0 : \R^d \rightarrow \R$ and a matrix $\mathbf{A} \in \R^{p \times d}$ with rank $d$ where $d \leq p$ such that $f_0( \x ) = g_0(\A \x)$.
\end{enumerate}We make assumption {\bf B7} for convenience; it can be loosened with sufficient algebraic footwork. Assumption {\bf B8} says that $f_0$ actually lives on a $d-$dimensional subspace; this is not restrictive as it is possible for $\A = I_p$. However, many situations arise when $d << p$. For example, there may be extraneous covariates or the mean function may be a function of a linear combination of the covariates---in effect, $d=1$. It is not required that $\A$ is known {\it a priori}, but simply that it exists. We also keep assumption {\bf B4}, which truncates the tails of the Gaussian priors for the hyperplane slopes and intercepts; this is done to bound the prior probability of the compliment of the sieve.

\begin{thm}\label{thm:fixedRate}
Assume that $\mathcal{X}$ is compact and that $f_0$ is convex, has continuous first order partial derivatives and suppose that conditions {\bf B1}, {\bf B4}, {\bf B7} and {\bf B8} hold. For both random covariates and fixed covariates and sufficiently large $V$,
\begin{equation}\notag
\p_{f_0}^{\infty}\,  \Pi\left( f \, : \, ||f - f_0||_n \geq H_n \epsilon_n  \g Y_1,\dots, Y_n , \x_1, \dots, \x_n \right)\rightarrow 0
\end{equation}for any $H_n \rightarrow \infty$, where $\epsilon_n^{-1} = \log( n) \, n^{1/(d+2)}$.
\end{thm}Theorem \ref{thm:fixedRate} is proven by showing that the conditions for Theorem 3 of \citet{GhVa07b} are satisfied. Details are given in Appendix B. 

We note that the rates achieved in Theorem \ref{thm:fixedRate} are within a log term of global minimax rates for general nonparametric convergence, $\epsilon_n = n^{-1/(p+2)}$, assuming $\A = I_p$. However, the $\epsilon$-metric entropy of the set of bounded convex functions with respect to the $|| \cdot ||_{\infty}$ metric scales like $\epsilon^{-p/2}$~\citep{VaWe96}, leaving open the possibility of convergence rates of $\epsilon_n = n^{-2/(p+4)}$ for bounded convex functions. In certain settings of convex-transformed density estimation that rate has been obtained~\citep{SeWe10}. We, however, do not believe that MBCR achieves this rate in a general setting.

\section{Implementation}\label{sec:implementation}
In this section, we extend MBCR to a model that can accommodate heteroscedastic data and provide a reversible jump MCMC sampler.

\subsection{Heteroscedastic Model}
The model in Section \ref{sec:model} assumes a global variance parameter, $\sigma^2$. While this is often a reasonable assumption, it can lead to particularly poor results when it is violated in a shape-restricted setting: locally chasing outliers in high-variance regions can lead to globally poor prediction due to the highly constrained nature of convex regression. 
To accommodate heteroscedasticity, we consider the following model,
\begin{equation}\notag
Y_i  = f(\x_i; \theta ) + \epsilon_i, \quad \epsilon_i  \sim N(0,g(\x_i)),
\end{equation}where $g : \R^p \rightarrow \R_+$. Specifically, to induce a flexible prior on $g$, we introduce a separate variance term for each hyperplane and modify the model to let
\begin{align}\notag
Y_i & = \max_{k \in \{1,\dots,K\}} \alpha_k + \beta^T_k \x_i + \epsilon_i, \quad \epsilon_i  \sim N(0,\sigma^2_k),\\\notag
(\theta_k,\sigma^2_k) & \sim N_{p+1}IG(\mu_{\alpha,\beta},V_{\alpha,\beta},a,b), \quad k = 1,\dots, K,\\\notag
K - 1 & \sim Poisson(\lambda).
\end{align}Here $N_{p+1}IG$ denotes the normal inverse gamma distribution with a $p+1$ dimensional normal. We choose a Poisson prior for the number of components, although we note that the model is generally not sensitive to the prior on the number of components. Due to the adaptable nature of the heteroscedastic model and its resistance to variance misspecification, we use it for all numerical work.

\subsection{Posterior Inference}
To sample from the posterior distribution, we use RJMCMC with the marginal posterior distribution of $\{K, \alpha, \beta, \sigma^2\}$ as the stationary distribution. Similar methods have been used for posterior inference on free-knot spline models by \citet{DeMaSm98} and \citet{DiGeKa01}. 

RJMCMC works by proposing a candidate model, $\{K^*, \alpha^*, \beta^*, {\sigma^2}^*\}$, and determining whether or not to move to that new model based on a Metropolis-Hastings type acceptance probability,
\begin{align}\label{eq:acceptanceP}
a(K^*, \alpha^*,\beta^*,{\sigma^2}^* & \g K, \alpha, \beta, \sigma^2) = \min\left\{ 1 , \frac{p(Y \g \x, K^*, \alpha^*, \beta^*,{\sigma^2}^*)}{p(Y \g \x, K, \alpha, \beta,\sigma^2)}\right.\\\notag
& \times \left. \frac{\Pi(K^*, \alpha^*, \beta^*,{\sigma^2}^*)}{\Pi(K, \alpha, \beta,\sigma^2)} \frac{q(K, \alpha, \beta,\sigma^2 \g K^*, \alpha^*, \beta^*, {\sigma^2}^*)}{q(K^*, \alpha^*, \beta^* , {\sigma^2}^* \g K, \alpha, \beta, \sigma^2)}\right\}.
\end{align}Here $p(Y \g \x, K^*, \alpha^*, \beta^*, {\sigma^2}^*)/p(Y \g \x, K, \alpha, \beta,\sigma^2)$ is the likelihood ratio of the data conditioned on the models, $\Pi(K^*, \alpha^*, \beta^*,{\sigma^2}^*)/\Pi(K, \alpha, \beta,\sigma^2)$ is the prior ratio of the models and $$q(K, \alpha, \beta, \sigma^2 \g K^*, \alpha^*, \beta^*, {\sigma^2}^*)/q(K^*, \alpha^*, \beta^*,{\sigma^2}^* \g K, \alpha, \beta,\sigma^2)$$ is an asymmetry correction for the proposal distribution. Candidate models are entirely new models:  all parameters are updated as a block. If only individual parameters or hyperplanes are updated, acceptance rates for parameters in the most constrained areas are orders of magnitude lower than those in the relatively unconstrained regions on the boundary of the function. Without block updates, there is poor mixing. There are three types of candidate models: hyperplane relocations, deletions and additions. All candidate models are generated from proposal distributions, which significantly impact the efficiency of the RJMCMC algorithm.

To generate proposal distributions we use the covariate partition induced by the current model $\{K,\alpha,\beta,\sigma^2\}$ to create a set of basis regions. Basis regions are determined by partitioning the set of training data. For example, suppose that a partition of the observations, $(\x_1,y_1),\dots,(\x_n,y_n)$, has $K$ subsets. Let $C = \{C_1,\dots,C_{K}\}$, where $
C_k = \left\{ i \, : \, i \mathrm{\ in \ subset \ } k \right\}.$ We can use $C$ to produce a set of basis regions for generating $(\alpha^*,\beta^*)$ with $K$ components,
\begin{align}\label{eq:basisRegion}
V_k^* & = \left(\tilde{V}_{\alpha,\beta}^{-1} + \x_{[k]}^T \x_{[k]}\right)^{-1},\\\notag
\mu_k^* & =  V_{k}^*\left(\tilde{V}_{\alpha,\beta}^{-1}\tilde{\mu}_{\alpha,\beta} + \x_{[k]}^T\y_{[k]}\right),\\\notag
a^*_k & =  \tilde{a} + \frac{n_k}{2},\\\notag
b^*_k & = \tilde{b} + \frac{1}{2}\left(\tilde{\mu}_{\alpha,\beta}^T \tilde{V}_{\alpha,\beta}^{-1} \tilde{\mu}_{\alpha,\beta} + \y_{[k]}^T \y_{[k]} - {\mu_{k}^*}^T {V_{k}^*}^{-1}{\mu_{k}^*}\right),\\\notag
(\alpha^*_k,\beta^*_k,{\sigma_k^2}^*) & \sim N_{p+1}IG\left(\mu_{k}^*,V_{k}^*, a_k^*,b_k^* \right), \quad k = 1,\dots,K.
\end{align}Here, $\x_{[k]} = \{[1, \x_i] \, : \, i \in C_k\}$, $\y_{[k]} = \{y_i \, : \, i \in C_k\}$ and $n_k$ is the number of elements in subset $k$. The hyperparameters for the proposal distributions, $(\tilde{\mu}_{\alpha,\beta},\tilde{V}_{\alpha,\beta},\tilde{a},\tilde{b})$, are not necessarily the same as those for the prior. Often the variance parameters are smaller to produce higher acceptance rates. The current set of hyperplanes, $\{K,\alpha,\beta\}$, are used to create the partitions that define the basis regions,
\begin{equation}\notag
C_k = \left\{ i \, : \, k = \arg \max_{j \in \{1, \dots, K\} } \alpha_j + \beta_{j}^T \x_i \right\}.
\end{equation}

For a relocation proposal distribution, the $K$ basis regions are generated by the covariate partition of the current model. The removal proposal distribution is a mixture with $K$ components. Each component is generated by removing the hyperplane $k$ for $k = 1,\dots, K$ and using the remaining $K-1$ hyperplanes to create a set of basis regions. Proposal distributions for additions are less straightforward. The addition proposal distribution is a mixture with $K L M$ components. Beginning with the subsets defined by the current model, $\{K,\alpha,\beta,\sigma^2\}$, each subset $j = 1,\dots,K$ is searched along a random direction $m$ for $m = 1,\dots, M$. On each of those random directions, the subset $j$ is divided according to a knot $a_{\ell}^j$ into the set of observations less than $a_{\ell}^j$ in direction $m$ and those greater. This is done for $\ell = 1,\dots, L$ knots for each subset $j$ and direction $m$. An example is shown in Figure \ref{fig:additions}. Full implementation details are given in Appendix C.

We note that the sampler for MBCR does not behave like a typical MCMC sampler. Convergence and mixing are extremely fast. Unlike most MCMC samplers, the MBCR sampler converges once the ``right'' number of components has been reached, typically within zero to four of the mean number of components. This is due to the way the proposal distributions are constructed and the strict requirements of convexity. Block updating ensures that autocorrelation drops to near zero rapidly. Numerical results suggest this generally happens after about three samples. While convexity endows the sampler with properties like fast convergence, it can also lead to situations where the restrictions are too rigid for the sampler to function. For example, if the noise level is very low, the number of observations is more than a few thousand, or the number of dimensions is moderate to high, the region of admissible models becomes very small and the acceptance rates rapidly drop to zero. Approximate inference methods seem to be required in these situations.

\begin{figure}[t]
\begin{center}
\includegraphics[width=5in,viewport=65 270 680 550]{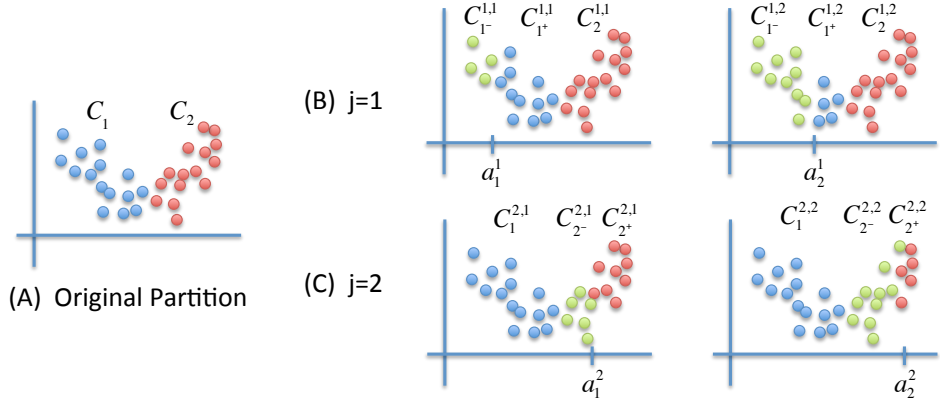}
\end{center}
\caption{Basis regions for one covariate combination when $L = 2$ and $K = 2$. (A) shows the original partition; (B) shows the partition when the region induced by the first hyperplane is split; (C) shows the partition when the region induced by the second hyperplane is split.}
\label{fig:additions}
\end{figure}


\section{Applications}\label{sec:numbers}
In Section \ref{sec:synthetic}, we compare the performance of MBCR to other regression methods on a set of synthetic problems. We show that convexity constraints can produce better estimates than their unconstrained counterparts and that MBCR is competitive with state of the art convex regression methods with respect to mean squared error. In Section \ref{sec:convex}, we analyze the behavior of MBCR, CAP and LSE when approximating an objective function for convex optimization. We show that MBCR produces estimates that are more suited to objective function approximation than those produced by CAP or LSE. 

\subsection{Synthetic Problems}\label{sec:synthetic}
In this subsection, we create a set of synthetic problems designed to show off the strength of convexity constraints. Problem 1 is highly non-linear and has moderate dimensionality (5); Problems 2 and 3 also have moderate dimensions in the covariate space (6 and 4, respectively), but both actually reside in a univariate subspace. 

\paragraph{Problem 1.}
Let $\x \in \R^5$.  Set $$y = \left(x_1+.5x_2+x_3\right)^2 - x_4+.25x_5^2 + \epsilon,$$
where $\epsilon \sim N(0,1)$. The covariates are drawn from a 5 dimensional standard Gaussian distribution, $N_5(0,I)$.

\paragraph{Problem 2.}
Let $\x \in \R^6$.  Set $$y = \left(x_1+x_2\right)^2 + \epsilon,$$
where $\epsilon \sim N(0,.5^2)$. The covariates are drawn from a 6 dimensional uniform distribution, $x_{j} \sim Unif[-1,1]$ for $j=1,\dots,6$.

\paragraph{Problem 3.}
Let $\x \in \R^4$.  Set 
\begin{align}\notag
y & = \left| \mathbf{a}^T\x \right| + \epsilon, & \mathbf{a}^T & = \left[0.8262,0.9305,1.6361,0.6072\right] 
\end{align}
where $\epsilon \sim N(0,1^2)$. The covariates are drawn from a 4 dimensional uniform distribution, $x_{j} \sim Unif[-4,4]$ for $j=1,\dots,4$.

\begin{table}
\caption{\label{tab:synthetic}Mean squared error on Problems 1, 2 and 3.}
\centering
\fbox{%
\begin{tabular}{l | r@{.}l | r@{.}l | r@{.}l | r@{.}l }
\multicolumn{9}{c}{Problem 1}\\\hline
Method & \multicolumn{2}{c |}{$n = 100$} & \multicolumn{2}{c |}{$n = 200$} & \multicolumn{2}{c |}{$n = 500$} & \multicolumn{2}{c }{$n = 1,000$}\\\hline
MBCR & {\bf1} &{\bf0373} & {\bf 0} & {\bf3679} & {\bf 0} & {\bf2784} & 0 & 2180 \\
 CAP & 1 & 6878 & 1 & 5336 & 0 & 3646 & {\bf 0} & {\bf1500}  \\
 LSE & 4 & 0174 & 1 & 4370 & 13 & 3398 & 1 & 8434 \\
 GP & 7 & 6612 & 6 & 2974 & 4 & 4793 & 3 & 5518  \\\hline
 \multicolumn{9}{c}{Problem 2}\\\hline
Method & \multicolumn{2}{c |}{$n = 100$} & \multicolumn{2}{c |}{$n = 200$} & \multicolumn{2}{c |}{$n = 500$} & \multicolumn{2}{c }{$n = 1,000$}\\\hline
MBCR & {\bf0} & {\bf 0943} & {\bf0} & {\bf0720} & {\bf0} & {\bf0155} &  0 & 0182 \\
 CAP & 0 & 1191 & 0 & 0887 & 0 & 0205 & {\bf 0} & {\bf 0129}  \\
 LSE & 4 & 6521 & 2 & 8926 & 1 & 9979 & 5 & 4998 \\
 GP & 0 & 3555 & 0 & 3932 & 0 & 3598 & 0 & 2174  \\\hline
 \multicolumn{9}{c}{Problem 3}\\\hline
Method & \multicolumn{2}{c |}{$n = 100$} & \multicolumn{2}{c |}{$n = 200$} & \multicolumn{2}{c |}{$n = 500$} & \multicolumn{2}{c }{$n = 1,000$}\\\hline
MBCR & {\bf 0 }& {\bf 1399} & {\bf 0} & {\bf 0775} & {\bf 0} & {\bf 0138} & {\bf 0} & {\bf 0102} \\
 CAP & 0 & 1886 & 0 & 1308 & 0 & 0192 & 0 & 0164  \\
 LSE & 4 & 7537 & 2 & 0210 & 1 & 4801 & 7 & 6638 \\
 GP & 2 & 0351 & 3 & 4649 & 2 & 9349 & 3 & 2026  \\
 \end{tabular}}
\end{table}

\subsubsection{Results.} 
On all of these problems, MBCR is compared to CAP, LSE and Gaussian Process priors~\citep{RaWi06}. Gaussian process priors are a Bayesian method that is widely and successfully used in regression and classification settings; we use the Matlab \texttt{gpml} package for implementation. All methods were implemented in Matlab; the least squares estimate (LSE) was found using the \texttt{cvx} optimization package. LSE took 5 to 6 minutes to run with 500 observations and 50 to 60 minutes to run with 1,000. The tolerance parameter for CAP was chosen through five-fold cross-validation. MBCR was implemented with component-specific variances. It was run for 1,000 iterations with the first 500 discarded as a burn-in.

Due to the highly constrained nature of the model and block updating, convergence of the sampler was extremely fast in all settings. The model was generally insensitive to the hyperparameter for the number of hyperplanes, $\lambda$; it was varied over three orders of magnitude and set to 20 for all tests. In lower, dimensions, however, choice of $\lambda$ was more important. Likewise, variance hyperparameters were tested over three orders of magnitude with little sensitivity. Distributions were not placed over the variance hyperparameters because of the delicate relationship between the proposal distributions and the hyperparameters. All mean hyperparameters were set to 0.

MBCR and CAP dramatically outperformed other methods on all of the problems. LSE had relatively poor performance although it includes convexity constraints. This is due to overfitting, particularly in boundary regions. MBCR and CAP performed comparably on Problems 2 and 3, which both reside in a univariate subspace of the general covariate space. However, MBCR outperformed CAP on the more complex Problem 1, particularly when there were few observations available.

\subsection{Objective Function Approximation for Stochastic Optimization}\label{sec:convex}

Stochastic optimization methods are used to solve optimization problems with uncertain outcomes. The traditional objective is to minimize expected loss. There are many problems in this class, ranging from stochastic search~\citep{Sp03} to sequential decision problems~\citep{SuBa98,Po07}. In this section, we study the use of convex regression to compute response surfaces. A response surface is an approximation of an objective function based on a collection of noisy samples. Once a response surface has been created, it is searched to estimate the minimizer or maximizer of a function. Convex representations are desirable. First, the resulting approximation will likely be closer to the true objective function than an unconstrained approximation. Second, and more importantly, the surrogate objective function is now convex as well and can be easily searched with a commercial solver.

Consider the following problem. We would like to minimize an unknown function $f(\x)$ with respect to $\x$ given $n$ noisy observations, $(\x_i,y_i)_{i=1}^n$, where $y_i = f(\x_i) + \epsilon_i$,  
\begin{equation}\label{eq:so}
\min_{\x \in \mathcal{X}} \E\left\{f(\x) \g (\x_i,y_i)_{i=1}^n\right\}.
\end{equation}

To solve Equation (\ref{eq:so}), we approximate $\E\left\{f(\x) \g (\x_i,y_i)_{i=1}^n\right\}$ with three different methods for regression: least squares, CAP and MBCR. Let $\hat{f}_n(\x)$ be the estimate of the mean function given $ (\x_i,y_i)_{i=1}^n$. Unlike CAP and LSE, MBCR is a Bayesian method; it places a distribution over functions rather than producing a single function estimate. Let $\hat{f}^{(m)}_n(\x)$ be a sample from the posterior; the Bayes estimate of the mean function can be approximated by the average of $M$ samples from the posterior,$$\hat{f}_n(\x) \approx \frac{1}{M} \sum_{m=1}^M  \hat{f}^{(m)}_n(\x).$$We demonstrate the empirical differences between the objective functions produced by MBCR, CAP and LSE by solving a small stochastic optimization problem.

\paragraph{Example.} Set
\begin{equation}\label{eq:optExample}
Y_i = \x_i \mathbf{Q} \x_i^T + \epsilon_i,\quad  \mathbf{Q} = 
\left[
\begin{array}{cc}
1  & 0.2   \\
 0.2 & 1  
\end{array}
\right], \quad \epsilon_i \sim N(0,0.1).
\end{equation}
The constraint set is $-1 \leq x_j \leq 1$ for $ j = 1,2$. Observations were sampled randomly from a uniform distribution, $\x_i \sim Uniform [-1,1]^2.$ We used LSE, CAP and MBCR to approximate the objective function. To examine the stability of these methods for objective function approximation, we sampled 100 observations 50 times for Equation (\ref{eq:optExample}). Approximations of the objective functions for one sample are shown in Figure \ref{fig:objective}.

\begin{figure}
\begin{center}
\includegraphics[width=3in,viewport= 70 200 560 650]{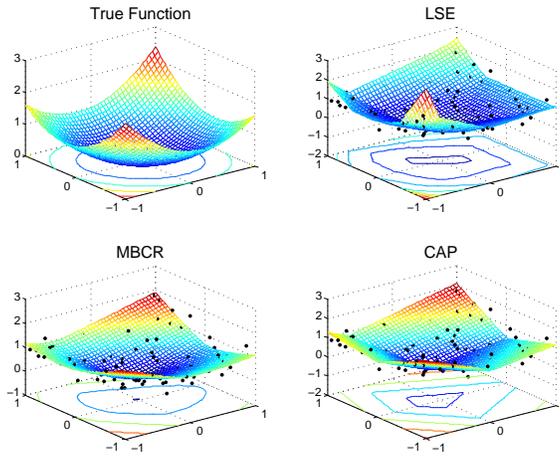}
\end{center}
\caption{Objective functions for LSE, CAP and MBCR for Equation (\ref{eq:optExample}) given 100 observations. Both LSE and CAP produce piecewise-linear functions; CAP produces a sparser function than LSE. MBCR averages over piecewise-linear functions to produce an estimate that is much closer to smooth.}
\label{fig:objective}
\end{figure}

\begin{figure}
\begin{center}
\includegraphics[width=3in,viewport= 110 220 500 580]{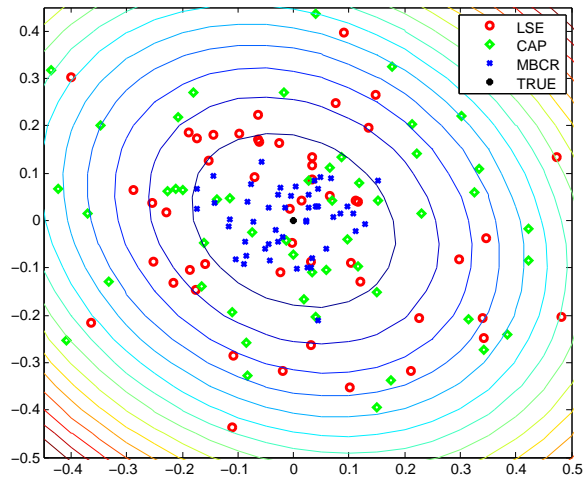}
\end{center}
\caption{Minima from the objective functions created by LSE, CAP and MBCR for Equation (\ref{eq:optExample}) given 100 observations; contours are from the true function. The observations were sampled 50 times; selections made when the objective function was approximated by MBCR are much more concentrated around the true minimum than those chosen using LSE or CAP.}
\label{fig:optMinima}
\end{figure}

\subsubsection{Results}We compared MBCR, CAP and LSE across 50 samples of 100 observations. The minima of piecewise planar models, like CAP and LSE, are on one of the vertices (or occasionally along one of the edges); this makes the minima of such models highly sensitive to model parameters such as number of hyperplanes and the value of their coefficients. MBCR, however, places a distribution over piecewise planar models. The Bayes estimate averages those models to produce something that is close to smooth and hence is relatively robust to observation design. Figure \ref{fig:objective} highlights these differences. The minima of both piecewise planar methods were sensitive to the observation design while the minima of MBCR proved more robust. Locations of minima are shown in Figure \ref{fig:optMinima}.

\subsubsection{Discussion}
Many methods for solving stochastic optimization problems, including response surface methods~\citep{BaMe06,Li10}, Q-learning~\citep{ErGeWe05} and approximate dynamic programming~\citep{Po07}, involve functional approximations that are then searched to find a solution that minimizes or maximizes the approximate reward. Current solution methods for these problems use either unconstrained regression methods~\citep{LaPa03,ErGeWe05} or additive approximations with univariate convex functions~\citep{PoRuTo04,NaPo09}. Robust multivariate convex regression methods could allow efficient solution of a broad set of stochastic optimization problems, inlcuding resource allocation, portfolio optimization and inventory management.

\section{Conclusions and Future Work}\label{sec:conclusions}
In this article, we introduced a novel fully Bayesian, nonparametric model for multivariate convex regression and showed strong posterior consistency along with convergence rates. We presented an efficient RJMCMC sampler for posterior inference. Our model was used to approximate objective functions for stochastic optimization and showed improvement over existing frequentist methods.

While this work represents a large advancement for convex regression, much remains to be done. First, we need to develop sampling methods that scale to large problems. Second, MBCR needs to be tested on a variety of stochastic optimization problems. Third, MBCR can be combined with other Bayesian methods to produce a class of semi-convex estimators.

Currently, the RJMCMC sampling method only scales well to moderate dimensionality and problem size:  its limits are about 8 to 10 dimensions and a few thousand observations. Approximate inference methods, such as variational Bayes, could allow MBCR to solve problems an order of magnitude larger. Implementation, however, is not a straightforward extension of existing methods. 

In stochastic optimization, MBCR is an extremely promising tool for value function approximation. Many solution methods for sequential decision problems include value function approximation, such as point-based value iteration~\citep{PiGoTh03}, fitted Q-iteration~\citep{ErGeWe05}, approximate dynamic programming~\citep{Po07}. All of these methods involve iterative searches of an approximate value function over sets of feasible actions where. In many problems, such as resource allocation, the value function is known to be convex. Robust multivariate convex regression methods would allow a wider variety of problems to be solved, including those with large action spaces and non-separable objective functions.

Perhaps the most intriguing feature of MBCR is that it is a Bayesian model and can easily be combined with other Bayesian models to produce estimators that are convex in some dimensions, but not all. For example, it is well known that consumer preferences for bundled products tend to be convex. However, it is likely that other covariates like consumer age, gender, income and education influence the preference function---and the function is not convex in these covariates. This set of functions could be well-modeled by a combination of MBCR and Bayesian mixture models like Dirichlet processes~\citep{Fe73,An74} or hierarchical Dirichlet processes~\citep{TeJoBe06}. Such flexible models would be of great value to an assortment of fields, including economics, operations research and reinforcement learning.

\section*{Acknowledgements}
This research was partially supported by grant R01ES17240 from the National Institute of Environmental Health Sciences (NIEHS) of the National Institutes of Health (NIH). Lauren A. Hannah is partially supported by the Duke Provost's Postdoctoral Fellowship.

\section*{Appendix A}\label{app:proofs}
Appendix A contains the proofs for Section \ref{sec:consistency}.

To show pointwise convergence, we use Theorems 1 to 3 of \citet{ChSc07}; they are condensed for this paper. For the fixed design case, let $Q_n$ be the empirical density of the design points, $Q_n(\x) = n^{-1} \sum_{i=1}^n \1_{\{\x_i\}}(\x).$ The empirical density is used to define the following neighborhood, $$W_{\epsilon,n} = \left\{ (f,\sigma) : \int \left| f(\x) - f_0(\x) \right| dQ_n(\x) < \epsilon, \, \left| \frac{\sigma}{\sigma_0} - 1 \right| < \epsilon\right\}.$$

\begin{thm}\label{thm:pointwiseFixed}\citep{ChSc07}
Let $\p_{f_0,\sigma_0}^{\infty}$ denote the joint conditional distribution of $\{Y_i\}_{i=1}^{\infty}$ given the covariates, assuming that $f_0$ is the true mean function and $\sigma_0^2$ is the true variance. If assumptions {\bf A1}, {\bf A2}, {\bf B1} and {\bf B2} are satisfied, then for every $\epsilon > 0$, $$\p_{f_0,\sigma_0}^{\infty} \Pi\left\{ (f,\sigma) \in W^C_{\epsilon,n} \, | \, Y_1,\dots,Y_n,\x_1,\dots,\x_n\right\} \rightarrow 0.$$
\end{thm}  

For the random design case, let $Q$ be the density of the random design points. Let $$U_{\epsilon} = \left\{ (f,\sigma) : \inf \{ \epsilon > 0 : Q(\{ \x : |f(\x) - f_0(\x)| > \epsilon \} ) < \epsilon \}, \,  \left| \frac{\sigma}{\sigma_0} - 1 \right| < \epsilon\right\}$$be the set of neighborhoods based on the in-probability metric.
\begin{thm}\label{thm:pointwiseRandom}\citep{ChSc07}
Let $\p_{f_0,\sigma_0}^{\infty}$ denote the joint conditional distribution of $\{Y_i\}_{i=1}^{\infty}$ given the covariates, assuming that $f_0$ is the true mean function and $\sigma_0^2$ is the true variance. If assumptions {\bf A1}, {\bf A2}, {\bf B1} and {\bf B2} are satisfied, then for every $\epsilon > 0$, $$\p_{f_0,\sigma_0}^{\infty} \Pi\left\{ (f,\sigma) \in U^C_{\epsilon} \, | \, Y_1,\dots,Y_n,\X_1,\dots,\X_n\right\} \rightarrow 0.$$
\end{thm}  

We now show that the prior satisfies assumptions {\bf A1} and {\bf A2} for Theorems \ref{thm:pointwiseFixed} and \ref{thm:pointwiseRandom}.

\begin{lem}\label{prop:priorKL}
For every $\delta > 0$, the prior $\Pi$ from {\bf B2} and {\bf B3} or {\bf B4} puts positive measure on the neighborhood $$B_{\delta} = \left\{ (f,\sigma^2) : || f - f_0 ||_{\infty} < \delta, \, \left| \frac{\sigma}{\sigma_0} - 1\right| < \delta \right\}.$$
\end{lem}
\begin{proof}Fix $\delta > 0$. Break $B_{\delta}$ into two parts, $\Pi(B_{\delta}(\beta)) = \left\{ f : || f - f_0 ||_{\infty} < \delta\right\}$,  and $\Pi(B_{\delta}(\sigma^2)) = \left\{\left| \frac{\sigma}{\sigma_0} - 1\right| < \delta \right\}.$ Under a truncated inverse gamma prior,
\begin{equation}\notag
\Pi(B_{\delta}(\sigma^2)) = \Pi((\sigma_0^2(1-\delta)^2,\sigma_0^2(1+\delta)^2)) > 0.
\end{equation}

To show prior positivity on $\Pi(B_{\delta}(\beta))$, we create a sufficiently fine mesh over $\mathcal{X}$. On each section of the mesh, we show that there exists a collection of hyperplanes that 1) do not intersect with $f_0$, and 2) have an $\ell_{\infty}$ distance from $f_0$ of less than $\delta$ in that section. Since $f_0$ is bounded on $\mathcal{X}$, it is Lipschitz continuous with parameter $L$.  A mesh size parameter $\gamma > 0$, which depends on $\delta$, can be found to make a $\gamma$ mesh over $\mathcal{X}$ with the following requirements.

Number regions $r = 1,\dots, R$; call the subsets of the covariate space defined by the regions $M_r^{\gamma}$. Because $f_0$ is Lipschitz continuous, an $\eta > 0$ can be found such that for every region $r$, one can find $\alpha_{r^*}$ and $\beta_{r^*}$ where for every $\alpha_r \in [\alpha_{r^*} - \eta, \alpha_{r^*} + \eta]$ and $\beta_r \in [\beta_{r^*} - \eta \1, \beta_{r^*} + \eta \1]$, 
\begin{align}\notag
		\alpha_r + \beta_r^T \x & < f_0(\x), & f_0(\x) - \alpha_r - \beta_r^T \x & < \delta
\end{align}for every $x \in M_r^{\gamma}$.
	
We create a function $f_{\delta}(\x)$ to approximate $f_0$ by taking the maximum over the set of $R$ hyperplanes; using the above, we can bound the distance between $f_0$ and $f_{\delta}$,
\begin{align}\notag
\sup_{\x \in \mathcal{X}} ||f_0(\x) - f_{\delta}(\x) ||_{\infty} & = \sup_{\x \in \mathcal{X}} ||f_0(\x) -\max_{r \in \{1,\dots,R\}} \alpha_r + \beta_r^T \x||_{\infty},\\\notag
& \leq \max_{r = 1,\dots,R} \sup_{\x \in M^{\gamma}_r} ||f_0(\x) - \alpha_r - \beta_r^T\x||_{\infty},\\\notag
& < \delta.
\end{align}
To complete the proof, we note that $\Pi(K = R) > 0$ and $\Pi([\alpha_{r^*} - \eta, \alpha_{r^*} + \eta], [\beta_{r^*} - \eta \1, \beta_{r^*} + \eta \1]) > 0 $ for $r = 1,\dots, R$.\qed \end{proof}

\begin{lem}\label{lem:priorSieve}
Define the prior $\Pi$ as in {\bf B2} and {\bf B3}. There exist constants $C_1>0$ and $c_1 > 0$ such that $\Pi(\Theta^C_n) \leq C_1 e^{-c_1n}$.
\end{lem}
\begin{proof}
Without loss of generality, assume $\mathcal{X} = [0,1]^p$. Note that 
\begin{align}\notag
\Theta_{1n}^c & = \Theta\setminus \left\{||f||_{\infty} < M_n, \, \left| \left| \frac{\partial}{\partial x_j} f \right| \right|_{\infty} < M_n, \, j = 1,\dots,p\right\},\\\label{eq:subset}
& \subseteq \bigcup_{k=1}^{\infty} \bigcup_{j=1}^k \bigcup_{\ell=1}^p \left\{f(\cdot; \theta) : K = k,  |\beta_{j,\ell}| \geq \frac{M_n}{2\sqrt{p}}\right\}\\\notag
& \quad \quad \quad \bigcup \left\{f(\cdot; \theta) : K = k,  |\alpha_{j}| \geq \frac{M_n}{2\sqrt{p}}\right\}.
\end{align}Taking the probability of the right hand side of Equation (\ref{eq:subset}),
\begin{align}\notag
\Pi(\Theta_{1n}^C) & \leq \sum_{k = 1}^{\infty} \Pi_{K} (K = k) \sum_{j=1}^k\left\{\Pi \left(|\alpha_j| \geq \frac{M_n}{2\sqrt{p}}\right)+ \sum_{\ell = 1}^p \Pi\left(|\beta_{j,\ell}| \geq \frac{M_n}{2\sqrt{p}}\right)\right\},\\\notag
& \leq 2 \E_{\Pi}[K] (p+1) \int_{c_0M_n}^{\infty}\frac{1}{\sqrt{2\pi }} e^{-\frac{1}{2} x^2}dx,\\\notag
& \leq C_1 e^{-c_1 n^{2\alpha}}. 
\end{align}\qed \end{proof}

Note that under the bounded prior assumption {\bf B4}, $\Pi(\Theta_n^c) = 0$ for sufficiently large $n$. Theorems \ref{thm:L1Fixed} and \ref{thm:L1Random} follow directly from Theorems 4 and 6, respectively, of \citet{ChSc07} and Theorems \ref{thm:pointwiseFixed} and \ref{thm:pointwiseRandom}. In the random design case, $L_1$ convergence is equivalent to in-probability convergence under assumptions {\bf B1} and {\bf B6}; the fixed design case requires more care. See \citet{ChSc07} for details.

\section*{Appendix B}\label{app:rateProofs}
Appendix B contains the proofs for Section \ref{sec:rate}. Theorem \ref{thm:fixedRate} relies on verifying the conditions of Theorem 3 of \citet{GhVa07b},
\begin{thm}[\citet{GhVa07b}]\label{thm:GhVa07b}
Let $\p_{f}^n$ be a product measure and $d_n$ a semimetric and let $\Theta$ be the space of all $\{K,\alpha,\beta\}$-tuples with positive measure under $\Pi$. Suppose that for a sequence $\epsilon_n \rightarrow 0$ such that $n\epsilon_n^2$ is bounded away from zero, all sufficiently large $j$ and sets $\Theta_n \subset \Theta$, the following conditions hold:
\begin{enumerate}[(i)]
	\item $\sup_{\epsilon > \epsilon_n} \log N(\epsilon/18, \{f \in \Theta_n: d_n(f,f_0) < \epsilon\},d_n) \leq n\epsilon_n^2;$
	\item There exist tests $\Phi_n$ such that $\E_{f_0}^n \Phi_n \leq e^{-\frac{1}{2}nd_n^2(f_0,f_1)}$ and $\E_{f}^n (1 - \Phi_n) \leq e^{-\frac{1}{2}nd_n^2(f_0,f_1)}$ for all $f \in \Theta$ such that $d_n(f,f_1) \leq \frac{1}{18}d_n(f_0,f_1);$
	\item $\frac{\Pi(\Theta_n^C) }{\Pi(B_n^*(f_0,\epsilon_n))} = o\left(e^{-2n\epsilon_n^2}\right);$
	\item $\frac{\Pi(f \in \Theta_n \, : \, j\epsilon_n < d_n(f,f_0) \leq 2j\epsilon_n) }{\Pi(B_n^*(f_0,\epsilon_n))} \leq e^{n\epsilon_n^2j^2/4},$
\end{enumerate}where $B_n^*(f_0,\epsilon_n) = \left\{ f \in \Theta \, : \, \frac{1}{n} \sum_{i=1}^n K_i(f_0,f) \leq \epsilon_n^2, \ \frac{1}{n} \sum_{i=1}^n V_i(f_0,f) \leq C \epsilon_n^2\right\}$. Then, $\p_{f_0}^{\infty} \Pi(f \, : \, d_n(f,f_0) \geq H_n \epsilon_n \g (X_i,Y_i)_{i=1}^n ) \rightarrow 0$ for every $H_n \rightarrow \infty$.
\end{thm}

The distance metric, $d_n$ that we will use is the $|| \cdot ||_{n}$ norm. Note that the $||\cdot||_n$ norm is bounded by the $||\cdot ||_{\infty}$ norm; we shall do metric entropy computations with respect to the $||\cdot ||_{\infty}$ norm. The values $K_i(f_0,f)$ and $V_i(f_0,f)$ denote $\int f_0 \log(f_0/f) d\mu$ and $\int f_0 (\log(f_0/f))^2 d\mu$, respectively. The quantity in condition {\it (i)} is the log of the covering number of the sieve under the supremum norm. To show that conditions {\it (i)} to {\it (iv)} of Theorem \ref{thm:GhVa07b} are met, we check them off one at a time while working in the linearly transformed space, $\tilde{\mathcal{X}} = \{ \mathbf{y} \, : \, \mathbf{y} = \A \x, \, \x \in \mathcal{X}\}.$

\begin{lem}\label{lem:metricEntropy}
Define $\Theta_n = \Theta_{1n}$ and suppose {\bf B4} holds. Then,
\begin{equation}\notag
\sup_{\epsilon > \epsilon_n} \log N(\epsilon/18, \{f \in \Theta_n: ||f-f_0||_{\infty} < \epsilon\},|| \cdot ||_{\infty}) \leq C \epsilon_n^{-d/2}.
\end{equation}
\end{lem}
\begin{proof}Working in the transformed space, assumption {\bf B4} places bounds on the supremum and partial derivatives for all $f \in \Theta$; the result then follows directly from Theorem 2.7.10 of \citet{VaWe96} or Theorem 6 of \citep{Br76} for $V = 1$. By setting $\tilde{\epsilon} = \epsilon/ V$, setting $\tilde{f} = f / V$, $\tilde{f}_0 = f _0/ V$ and calculating the metric entropy with respect to $\tilde{\epsilon}$, $\tilde{f}$ and $\tilde{f}_0$, the result holds. This covering needs to be repeated at most $\epsilon^{-p}$ times to cover the original space; taking the log, $p\log(1/\epsilon)$ can be bounded by a constant times $\epsilon^{-d/2}$.\qed \end{proof}

\begin{lem}\label{lem:piBounds}
Define $\Pi$ by {\bf B4} and {\bf B7}. Let $$B_n^*(f_0,\epsilon_n) = \left\{ f \in \Theta \, : \, \frac{1}{n} \sum_{i=1}^n K_i(f_0,f) \leq \epsilon_n^2, \ \frac{1}{n} \sum_{i=1}^n V_i(f_0,f) \leq C \epsilon_n^2\right\}.$$ Then there exist $C_1$ and $c_1>0$ such that
\begin{equation}\notag
\Pi\left(B_n^*(f_0,\epsilon_n) \right) \geq C_1 e^{c_1 \epsilon_n^{-d} \log \epsilon_n}.
\end{equation}
\end{lem}
\begin{proof}
By simple calculations,
\begin{align}\notag
K_i(f_0,f) & = \frac{1}{2\sigma_0^2} \left(f_0(\x_i)-f(\x_i)\right)^2, & V_i(f_0,f) & = \frac{1}{\sigma_0^2} \left(f_0(\x_i)-f(\x_i)\right)^2.
\end{align}To place a lower bound on the prior measure of $B_n^*(f_0,\epsilon_n),$ we construct a subset and place prior bounds on that.

Let $\beta \in \R^p$; a truncated Gaussian prior on $\beta$ induces a truncated Gaussian prior on $\tilde{\beta} = \A\beta$, the slope parameters in the transformed space. WLOG, take $\tilde{\mathcal{X}} = [0,1]^d.$ Set $\delta = \frac{1}{8\sqrt{d}\sigma_0^2}\epsilon_n$; let $\y_1,\dots, \y_m$ be a $\delta-$net over $\tilde{\mathcal{X}}$. The net can be chosen such that $m \leq K_m / \epsilon_n^d$ for some constant $K_m$ that depends only on $d$ and $\sigma_0^2$. Let $$\left(\alpha_{k}^*,\tilde{\beta}_{k,1}^*,\dots,\tilde{\beta}_{k,d}^*\right)= \left(g_0(\y_k),\frac{\partial}{\partial{x_1}}g_0(\y_k),\dots,\frac{\partial}{\partial{x_d}}g_0(\y_k)\right).$$Then with a sufficiently large truncation parameter $V$, for every $k \in \{1,\dots,m\}$,
\begin{align}\notag
\Pi_{\tilde{\theta}}\left((\alpha_k,\beta_{k,1},\dots,\beta_{k,d}) \in \left(\alpha_{k}^*,\tilde{\beta}_{k,1}^*,\dots,\tilde{\beta}_{k,d}^*\right) \pm \frac{1}{8\sigma_0^2}\epsilon_n \right) & \geq K_a \epsilon_{n}^{d+1},
\end{align} for some $K_a >0$ that depends on $d$, $\sigma_0^2$, $\A$ and $g_0$. Set $g(\y) = \max_{k \in \{1,\dots,m\}} \alpha_k + \beta^T_k \y.$ Then, $ \frac{1}{2\sigma_0^2} \left(f_0(\x_i)-g(\tilde{\x}_i)\right)^2  \leq \epsilon_n^2,$ so
\begin{align}\notag
 \Pi\left(B_n^*(f_0,\epsilon_n) \right) & \geq \Pi_K(K= m) \\\notag
 & \quad \times \sum_{k=1}^m \Pi_{\tilde{\theta}}\left((\alpha_k,\beta_{k1},\dots,\beta{k,d}) \in \left(\alpha_{k}^*,\tilde{\beta}_{k,1}^*,\dots,\tilde{\beta}_{k,d}^*\right) \pm \frac{1}{8\sigma_0^2}\epsilon_n \right),\\\notag 
 & \geq C_1 e^{c_1 \epsilon_n^{-d} \log \epsilon_n},
 \end{align}for some constants $C_1, c_1 > 0$.\qed \end{proof}
 
 We can use Lemma \ref{lem:piBounds} to check conditions {\it (iii)} and {\it (iv)} of Theorem \ref{thm:GhVa07b}.
 \begin{lem}\label{lem:conds3and4}
 Define $\Pi$ by {\bf B4} and {\bf B7}. Then for every large $j$,
 \begin{align}\notag
 \frac{\Pi(\Theta_n^C) }{\Pi(B_n^*(f_0,\epsilon_n))} & = C_2 e^{-c_2n - c_1 \epsilon_n^{-d} \log \epsilon_n},\\\notag
 \frac{\Pi(f \in \Theta_n \, : \, j\epsilon_n < ||f,f_0||_{\infty} \leq 2j\epsilon_n) }{\Pi(B_n^*(f_0,\epsilon_n))} & \leq C_1 e^{- c_1 \epsilon_n^{-d} \log \epsilon_n}.
 \end{align}
\end{lem}
\begin{proof}
The first equation can be bounded by using Lemma \ref{lem:priorSieve}; the second by setting the numerator equal to 1. \qed
\end{proof}

Now we use this collection of Lemmas and Theorem \ref{thm:GhVa07b} to prove Theorem \ref{thm:fixedRate}.
\begin{proof}[Proof of Theorem \ref{thm:fixedRate}]
We begin by checking the conditions of Theorem \ref{thm:GhVa07b}. Condition {\it (i)} follows from Lemma \ref{lem:metricEntropy}. Setting $\epsilon_{n}^{-1} = \log(n) \, n^{1/(d+2)}$, conditions {\it (iii)} and {\it (iv)} follow from Lemma \ref{lem:conds3and4}. Finally, \citet{Birge06} shows that the likelihood ratio test for $f_0$ versus $f_1$ satisfies condition {\it (ii)} relative to the $||\cdot||_n$ norm under both fixed and random design. Therefore, the main result follows directly from Theorem \ref{thm:GhVa07b}.\qed \end{proof}

\section*{Appendix C}\label{sec:inference}
The RJMCMC algorithm is similar to the ones proposed by \citet{DeMaSm98} and \citet{DiGeKa01} for BARS. Jumps in the chain can take three forms: additions, deletions and relocations. The probabilities of additions, deletions and relocations must satisfy detailed balance equations,
\begin{align}\label{eq:detailedBalance}
& \Pi (K+1,\alpha_{1:K+1}^*,\beta_{1:K+1}^*,{\sigma^2_{1:K+1}}^*) \\\notag 
& \quad \times p(K,\alpha_{1:K},\beta_{1:K}, {\sigma^2_{1:K}} \g K+1,\alpha_{1:K+1}^*,\beta_{1:K+1}^*,{\sigma^2_{1:K+1}}^*) \\\notag
& = \Pi(K,\alpha_{1:K},\beta_{1:K},\sigma^2_{1:K}) p(K+1,\alpha_{1:K+1}^*,\beta_{1:K+1}^*,{\sigma^2_{1:K+1}}^* \g K,\alpha_{1:K},\beta_{1:K},\sigma^2_{1:K}).
\end{align}\citet{DiGeKa01} shows that Equation (\ref{eq:detailedBalance}) is satisfied if additions, deletions and relocations are attempted with the following probabilities, respectively,
\begin{align}\notag
b_k & = c \min \left\{1, \frac{p(k+1)}{p(k)}\right\},& d_k & = c \min \left\{1, \frac{p(k-1)}{p(k)}\right\}, & r_k & = 1 - b_k - d_k,
\end{align}where $p(k)$ is the prior probability of $k$ hyperplanes and $c$ is a constant; we set $c = 0.4$.

\paragraph{Additions.}Given the current state $(K, \alpha, \beta)$, a new state with $K+1$ hyperplanes, $(K+1, \alpha^*,\beta^*)$, is proposed with the jump probability,
\begin{equation}\notag
q(K+1,\alpha_{1:K+1}^*,\beta_{1:K+1}^*,{\sigma^2_{1:K+1}}^*\g K,\alpha_{1:K},\beta_{1:K},\sigma^2_{1:K}) = b_K h_b(\alpha^*,\beta^*,{\sigma^2}^*\g \alpha, \beta,\sigma^2).
\end{equation}Here $b_K$ is the addition probability given $K$ hyperplanes and $h_b$ is the proposal distribution for additions.

\paragraph{Deletions.} A new state with $K-1$ hyperplanes, $(K-1,\alpha^*,\beta^*)$ is proposed with the jump probability,
\begin{equation}\notag
q(K-1,\alpha_{1:K-1}^*,\beta_{1:K-1}^*,{\sigma^2_{1:K-1}}^* \g K, \alpha_{1:K},\beta_{1:K},\sigma_{1:K}^2) = d_K h_d(\alpha^*,\beta^*,{\sigma^2}^*\g \alpha,\beta,\sigma^2).
\end{equation}Here $d_K$ is the deletion probability given $K$ hyperplanes and $h_d$ is the proposal distribution for deletions.

\paragraph{Relocations.} 
A new state with $K$ hyperplanes, $(K,\alpha^*,\beta^*)$ is proposed with the jump probability,
\begin{equation}\notag
q(K, \alpha_{1:K}^*,\beta_{1:K}^*,{\sigma^2_{1:K}}^*\g K, \alpha_{1:K},\beta_{1:K},\sigma_{1:K}^2) = r_K h_r(\alpha^*,\beta^*,{\sigma^2}^* \g \alpha, \beta,\sigma^2).
\end{equation}Here $r_K$ is the relocation probability given $K$ hyperplanes and $h_r$ is the proposal distribution for relocations. The full RJMCMC algorithm is given in Algorithm \ref{alg:RJMCMC}.

\begin{algorithm}[t]
\caption{Reversible Jump MCMC for MBCR}
\label{alg:RJMCMC}
\begin{algorithmic}
	\STATE Initialize $(K,\alpha,\beta,\sigma^2)$:  set $K = 1$, draw $(\alpha_1, \beta_1,\sigma^2_1)$ from posterior
	\LOOP
		\STATE Draw a new $(K^*,\alpha^*,\beta^*,{\sigma^2}^*)$ from the proposal distribution
		\STATE Set $(K,\alpha,\beta,\sigma^2)$ to $(K^*,\alpha^*,\beta^*,{\sigma^2}^*)$ with probability $a(K^*, \alpha^*,\beta^* , {\sigma^2}^*\g K, \alpha, \beta, \sigma^2)$
	\ENDLOOP
\end{algorithmic}
\end{algorithm}

\subsection*{Proposal Distributions}
Posterior inference for our model is particularly sensitive to the choice of proposal distributions. The space of potential hyperplanes is quite large and grows with $p$. In order to efficiently search that space, we use a collection of {\it basis regions} to create $h_b$, $h_d$ and $h_r$. Basis regions are determined by partitioning the set of training data, as described in Equation (\ref{eq:basisRegion}). We show how this is done for relocations, deletions and additions.

\paragraph{Relocations.} In a relocation step, the proposal distribution is generated by the basis regions created by the current $(\alpha,\beta)$. That is, $C^r = \{C^r_1,\dots,C^r_K\}$, where
\begin{equation}\notag
C_k^r = \left\{ i \, : \, k = \arg \max_{j=\{1,\dots,K\}} \alpha_j + \beta_j^T \x_i\right\}.
\end{equation}Then the proposal distribution is created using this partition as in Equation (\ref{eq:basisRegion}).

\paragraph{Deletions.} In a deletion step, the proposal distribution is a mixture of distributions generated by basis regions,
\begin{equation}\notag
h_d(\alpha^*,\beta^*,{\sigma^2}^*\g \alpha,\beta,\sigma^2) = \sum_{j=1}^K p_d(j) h_{d}^*(\alpha^*,\beta^*,{\sigma^2}^* \g \alpha_{-j},\beta_{-j},\sigma_{-j}^2),
\end{equation}where $\sum_{j=1}^K p_d(j) = 1$, and the $-j$ subscript denotes the indices $\{1,\dots,j-1,j+1,\dots,K\}$. For $j = 1,\dots, K$, the distribution $h_d^*(\alpha^*,\beta^*,{\sigma^2}^* \g \alpha_{-j},\beta_{-j},{\sigma}_{-j}^2)$ is made by creating the following partition, $C^{d,j} = \{C^{d,j}_1,\dots,C^{d,j}_{K-1}\}$, where
\begin{equation}\notag
C_k^{d,j} = \left\{ i \, : \, k = \arg \max_{\ell=\{1,\dots,j-1,j+1,\dots,K\}} \alpha_{\ell} + \beta_{\ell}^T \x_i\right\}.
\end{equation}Likewise, the distribution $h^*_d(\alpha^*,\beta^*,{\sigma^2}^* \g \alpha_{-j},\beta_{-j},\sigma_{-j}^2)$ is defined by Equation (\ref{eq:basisRegion}). The component probability $p_d(j)$ is set to be proportional to $1/|C^r_j|$, the inverse of the number of components supported by hyperplane $j$. If no components are currently supported, we set $p_d(j) \propto 1/.25$.

\paragraph{Additions.} As in the deletion step, the proposal distribution is a mixture of distributions generated by basis regions. Unlike deletions, it is not obvious how to add a hyperplane in a way that will result in a high quality proposal. In an addition step, MBCR starts with a set of hyperplanes, $(\alpha,\beta)$, and adaptively adds an additional hyperplane in the following manner. 
The current hyperplanes $(\alpha, \beta)$ define a partition over the observation space, $C = \{C_1,\dots,C_K\}$, where
\begin{equation}\notag
C_k = \left\{ i \in \{1,\dots,n\} \, : \, k = \arg \max_{\{k = 1,\dots,K\}} \alpha_k + \beta_{k}^T \x_i\right\}.
\end{equation}MBCR splits each element $ j = 1,\dots,K$ of $C$ in turn along a direction defined by a linear combination of covariates, producing a collection of new covariate partitions. A number of knots, $L$, is chosen {\it a priori}, along with $M$ random linear combinations, $(g_1^{m},\dots,g_{p}^m)_{m=1}^M$. Then, for each direction $m = 1,\dots, M$, and each knot $\ell = 1,\dots, L$, a covariate partition is generated in the following manner,
\begin{align}\notag
C_k^{b, j,\ell,m} & =  \left\{ i \, : \, k = \arg \max_{\{k = 1,\dots,K\}} \alpha_k + \beta_{k}^T \x_i, \ j \neq k \right\},\\\notag
C_{j^-}^{b, j,\ell,m} & = \left\{ i \, : \, j = \arg \max_{\{k = 1,\dots,K\}} \alpha_k + \beta_{k}^T \x_i, \ {\mathbf{g}^m}^T\x_j \leq a_{\ell}^j\right\},\\\notag
C_{j^+}^{b, j,\ell,m} & = \left\{ i \, : \, j = \arg \max_{\{k = 1,\dots,K\}} \alpha_k + \beta_{k}^T \x_i, \ {\mathbf{g}^m}^T \x_j > a_{\ell}^j\right\},
\end{align}where $a_{\ell}^j$ are chosen to produce $L+1$ intervals between $\min\{{\mathbf{g}^m}^T \x_{i} \, : \, i \in C_j\}$ and $\max\{{\mathbf{g}^m}^T \x_{i} \, : \, \ i \in C_j \}$. Set $$C^{b,j,\ell,m} = \left\{C_1^{b,j,\ell},\dots,C^{b,j,\ell,m}_{j-1},C^{b,j,\ell,m}_{j^-},C^{b,j,\ell,m}_{j^+},C^{b,j,\ell,m}_{j+1}\dots,C^{b,j,\ell,m}_{K}\right\}.$$
Often, it is convenient to chose the cardinal directions, that is, $\mathbf{g}^M = e_j$, as the linear combinations for the covariates. However, when $p$ is large and a sparse underlying structure is assumed, it is useful to choose $\mathbf{g}^m$ to be a random Gaussian vector with $M < p$.

The observation partitions are used to produce the following mixture model for the addition proposal distribution,
\begin{equation}\notag
h_b(\alpha^*,\beta^*,{\sigma^2}^* \g \alpha, \beta,\sigma^2) = \sum_{j=1}^p \sum_{\ell=1}^L p_b(j,\ell,m) h^*_b(\alpha^*,\beta^*,{\sigma^2}^* \g C^{b,j,\ell,m}),
\end{equation}where the distribution $h^*_b(\alpha^*,\beta^*,{\sigma^2}^* \g C^{j,\ell,m})$ is defined Equation (\ref{eq:basisRegion}). The weights $p_b(j,\ell,m)$ are set to be proportional to
\begin{equation}\notag
p_b(j,\ell,m) \propto n_{j^-}^{j,\ell,m}n_{j^+}^{j,\ell,m},
\end{equation}where $n_{j^-}^{j,\ell,m}= |C_{j^-}^{b,j,\ell,m}|$ and $n_{j^+}^{j,\ell,m}= |C_{j^+}^{b,j,\ell,m}|$. This gives higher weight to partitions that split a large number of observations fairly evenly.

\bibliographystyle{agsm}
\bibliography{../../BibFiles/localRefs}
\end{document}